\documentclass[11pt]{article}

\usepackage{algorithmic}
\usepackage{algorithm}
\usepackage{amssymb}
\usepackage{amsthm}
\usepackage{color}
\usepackage[text={6.5in,9.0in},centering]{geometry}
\usepackage{graphicx}
\usepackage{subfig}
\usepackage{hyperref}
\usepackage{makeidx}  
\usepackage{latexsym}
\usepackage{amsmath}
\usepackage{cases}
\usepackage{amsfonts}
\usepackage{epsfig}
\usepackage{boxedminipage}
\usepackage{url}

\newtheorem{result}{Result}

\newtheorem{theorem}{Theorem}
\newtheorem{lemma}{Lemma}

\newtheorem{definition}{Definition}

\newcommand{\no}{\noindent}
\newcommand{\bea}{\begin{eqnarray}}
\newcommand{\eea}{\end{eqnarray}}
\newcommand{\beq}{\begin{equation}}
\newcommand{\eeq}{\end{equation}}
\newcommand{\beqs}{\begin{equation*}}
\newcommand{\eeqs}{\end{equation*}}
\newcommand{\beas}{\begin{eqnarray*}}
\newcommand{\eeas}{\end{eqnarray*}}
\newcommand{\ep}{\end{proof}}
\newcommand{\bp}{\begin{proof}}

\def\cS{\mathcal{S}}
\def\wg{\widehat{g}}
\def\wj{\hat{j}}
\def\wt{\hat{t}}
\def\wi{\hat{i}}

\def\olta{\bar{L}_{j,k}^{t_a}}

\def\fjkita{f_{j,k}^{i,t_a}}
\def\gjkita{g_{j,k}^{i,t_a}}

\def\fjkiqta{f_{j,k}^{i,t_a}}
\def\gjkiqta{g_{j,k}^{i,t_a}}

\def\fjkit1{f_{j,k}^{i,t_1}}
\def\fjkit1{f_{j,k}^{i,t_1}}
\def\fjktt1{f_{j,k}^{t,t_1}}
\def\gjktt1{g_{j,k}^{t,t_1}}
\def\gjktt1{g_{j,k}^{t,t_1}}

\def\fjkxut1star{f_{j^*,k^*}^{x,u,t^*_1}}

\def\fjkita{f_{j,k}^{i,t_a}}
\def\gjkita{g_{j,k}^{i,t_a}}
\def\ljkita{\lambda_{j,k}^{i,t_a}}
\def\hljkita{\widehat{\lambda}_{j,k}^{i,t_a}}
\def\hWjkita{\widehat{W}_{j,k}^{i,t_a}}

\def\Pjktt1{P_{j,k}^{t,t_1}}

\def\PBS0{\frac{\widehat{h}_b P_T}{M} }

\def\PBS{P^{b^*}  }

\def\-{  \!-\! }
\def\+{  \!+\! }

\newcommand{\remove}[1]{}

\title{The Price of Anarchy is Unbounded for Congestion Games with Superpolynomial Latency Costs}

\author{ Rajgopal Kannan$^\dagger$,  Costas Busch$^\dagger$ and 
Paul Spirakis$^\ddagger$ \bigskip \\ 
$^\dagger$School of EECS, Louisiana State University, Baton Rouge, LA 70803\\ \{rkannan,busch\}@csc.lsu.edu \\
$^\ddagger$Dept. of Computer Engineering and Informatics, Univ. of Patras, Greece\\
spirakis@cti.gr
}

\begin{document}
\begin{titlepage}

\date{}
\maketitle

\thispagestyle{empty}

\begin{abstract}
We consider non-cooperative unsplittable congestion games where players share resources, and each player's strategy is pure and consists of a subset of the resources on 
which it applies a fixed weight.  Such games represent unsplittable routing flow games and also job allocation games. 
The congestion of a resource is the sum of the 
weights of the players that use it and the player's cost function is the sum 
of the utilities of the resources on its strategy.
The social cost is the total weighted sum of the player's costs.
The quality of Nash equilibria is determined by the price of anarchy ($PoA$) 
which expresses how much worse is the social outcome 
in the worst equilibrium versus the optimal coordinated solution.
In the literature the predominant work has only been on games with
polynomial utility costs, where it has been proven that the price of anarchy is bounded 
by the degree of the polynomial.
However, no results exist on general bounds for non-polynomial utility functions.

Here, we consider general versions of these  games in which the utility of each resource is an 
arbitrary non-decreasing function of the congestion. In particular, we consider a large family of 
superpolynomial utility functions which are asymptotically larger than any polynomial. We
demonstrate that for every such function there exist games for which the price of anarchy is unbounded
and increasing with the number of players (even if they have infinitesimal weights) while network resources 
remain fixed. We give tight lower and upper bounds which show this dependence on the number of players. 
Furthermore we provide an exact characterization of the $PoA$ of all congestion games whose
utility costs are bounded above by a polynomial function. Heretofore such results existed only for games with polynomial cost functions.

\end{abstract}
\end{titlepage}

\section{Introduction}
We consider non-cooperative congestion games on a set of resources which are shared among players.
A player's strategy consists of a subset (or all) of the resources
where it applies a fixed weight.
Strategies of players are pure in the sense that each player picks one strategy among a set of available strategies.
The resource utilization is unsplittable within a strategy, since a player applies the same weight on each resource.
The congestion of a resource is simply the sum of the weights of the players 
that use it.
The utility of each resource is a function of its congestion. 
Each player selfishly minimizes its own cost which is the sum of the utilities 
of all the resources along its strategy.

We examine pure Nash equilibria which are game states 
where each player has chosen a locally optimal strategy 
from which there is no better alternative strategy.
Rosenthal \cite{rosenthal1} shows that if all players have the same weight
then pure Nash equilibria always exist.
There may be multiple Nash equilibria for the same game.
The quality of a Nash equilibrium is measured with respect to the social cost function
which is simply the weighted sum of the all the players' costs. 
We measure the impact of the selfishness of the players with the {\em price of anarchy}
which is the ratio of the social cost of the worst Nash equilibrium versus the coordinated optimal social cost.

Resource congestion games can represent network flow games and job distribution games.
In networks, each resource corresponds to a link. A player with weight $w$ represents a routing request from a source node to a destination node of demand $w$ which is fulfilled along a path of the network. The utility cost of the player relates to the delay for sending the demand in the network along the chosen route. In job distribution games, each resource represents a machine. Each player has a job that consists of small sub-tasks that can execute at each node. The weight of the player $w$ relates to the work to be assigned to each machine in order to execute the job. The cost of the player relates to the delay to finish its job. In both the network and job games,
the price of anarchy represents the impact of selfishness to the overall performance of the system,
which in one case is the total network delay, and in the other the total work to execute all jobs.

Congestion games were introduced and
studied in~\cite{monderer1,rosenthal1}.
Most of the literature considers linear or polynomial utility functions.
Koutsoupias and Papadimitriou \cite{KP99} introduced
the notion of price of anarchy
in the specific {\em parallel link networks} model
in which they provide the bound $PoA = 3/2$.
Roughgarden and Tardos \cite{roughgarden3}
provided the first result for splittable flows in general networks
in which they showed that $PoA\le 4/3$.
Pure equilibria with atomic flow have been studied in
\cite{AAE05,BM06,SAGT,CK05,libman1,STZ04} for classic congestion games and their variations of botteleck games (where the cost is determined by the maximum congested edge),
and with splittable flow in
\cite{roughgarden1,roughgarden2,roughgarden3,roughgarden5}.
Mixed equilibria with atomic flow have been studied in
\cite{czumaj1,GLMMb04,KMS02,KP99,LMMR04,MS01,P01},
and with splittable flow in
\cite{correa1,FKS02}.

\section{Contributions: Functional Characterization of PoA Boundedness}

Let $\mathcal{L} = \{l_1(x), l_2(x), \ldots \}$ denote a class
of arbitrary non-decreasing latency cost functions. 

\begin{definition}
Function $l_k(x) \in \mathcal{L}$ is defined to be a superpolynomial function if it cannot be bounded from above by some polynomial function $x^p$ i.e 
\(
\not \exists p: 
\lim_{x \to \infty} \frac{l_k(x)}{x^p} \rightarrow 0
\)
\label{superpolynomial}
\end{definition}

We also define our notion of boundedness for the Price of Anarchy.

\begin{definition}
The Price of Anarchy of a congestion game $G$ is bounded if it does not arbitrarily increase with the number of players.
\label{boundedness}
\end{definition}

Under our notion of boundedness, the Price of Anarchy depends only on intrinsic game parameters such as
the parameters of the latency cost function, player strategies etc. but is {\it independent} 
of the number of players. Consider for example a game with 2 resources and $n$ players who can use either resource.
The $PoA$ is 1 and independent of $n$. In this paper, we prove the existence of a large class of latency 
functions for which there exist games in which  the $PoA$ increases with the number of players while other network 
parameters such as network topology or number of resources, player weights and cost functions remain fixed. 

Let $G$ be an unsplittable congestion game with player weights derived from weight set $W \subseteq \mathbb{R}^+$ and latency cost 
functions derived from $\mathcal{L}$. We assume that $W$ is bounded by $w = (\max_i \in W)$ representing the weight of the largest 
player. For any function $l_k() \in \mathcal{L}$, define the set of ordered triples 

\beq
O_k = \left \{(x,y,z)_k \in \mathcal{R}^3 |  x \geqslant y \geqslant z, 0 < z \leqslant w, \frac{l_k(x+z)}{l_k(x)} = \frac{x}{y} \right \}
\label{orderedtriples}
\eeq

As per the usual convention, for any ordered triple $(j,t,i) \in \mathcal{R}^3$, we denote $(j,t,i) \geqslant (x,t,i)_k$ if $j \geqslant x$. 
Next, for each $l_k()$, we define two special parameters:



\bea
g^*_k &= &\max\limits_{(x,y,z)_k \in O_k} \left \{ \frac{l_k(x+z)}{l_k(y)} \right \} 
\label{g*k}
\\
\wg_k &= &\max\limits_{x \in \mathbb{R}^+, y \geqslant z}
\frac{l_k(x+z)}{l_k(y)} - \frac{x l_k(x)}{y l_k(y)}
\label{ghatk}
\eea

\no Let the ordered triple values at $g^*_k$ and $\wg$ be denoted
by $(j^*,t^*,i^*)_k$ and $(\wj, \wt, \wi)_k$, respectively,
i.e $g^*_k = l_k(j^* + i^*)/l_k(t^*)$ and  $\wg_k =
\frac{l_k(\wj+\wi)}{l_k(\wt)} - \frac{\wj l_k(\wj)}{\wt l_k(\wt)}$.
Note that  $(\wj,\wt,\wi) \leqslant (j^*,t^*,i^*)$ since both $l_k(x+z)$
and $x l_k(x)$ are increasing in $j$ and hence by definition of the
ordered triple $(x,y,z)_k$ we must have $\widehat{g}_k \leqslant g^*_k$.

We consider three disjoint subclasses of latency functions from $\mathcal{L}$ as described below and 
evaluate their $PoA$ bounds.


\begin{itemize}

\item 
$\mathcal{L}_1 = \{ l_k \in \mathcal{L} | \forall i \in \mathbb{R}^+
\lim_{x \to \infty}  \frac{l_k(x+i)}{l_k(x)} > 1 \}$. Thus $\mathcal{L}_1$ contains superpolynomial functions such as $l_k(x) = a_k 2^x$, $l_k(x) = a_k x^x$ etc. with coefficient $a_k>0$.

\item 
$\mathcal{L}_2 = \{ l_k \in \mathcal{L} | l_k(x) = a_k \cdot e^{\log^{1+\epsilon}x} \}$, where $\epsilon >0$ is any constant, coefficient $a_k >0$ and $\log$ refers to the natural logarithm. The functions in $\mathcal{L}_2$ are superpolynomials with the property $\lim_{x \to \infty}  \frac{l_k(x+i)}{l_k(x)} = 1$. Thus $\mathcal{L}_2$ contains slower growing superpolynomial functions such as $l_k(x) = a_k x^{\log x}$, $l_k(x) = a_k x^{\log^2 x}$ etc.

\item 
$\mathcal{L}_3$ is the class of non-superpolynomial increasing functions inclusive of and bounded from above by polynomials, for example, $l_k(x) = (\sum_{q=0}^d a_q x^q)(\sum_{q=0}^d b_q \log^q x)$.

\end{itemize}

Our first result  is an exact bound on the Price of Anarchy of unsplittable congestion games.

\begin{result}
For every  subset of cost functions $\mathcal{L}' \subseteq \in \mathcal{L}$, there exist games $G$
where the Price of Anarchy is 

\beq
\displaystyle
PoA(G) = 
\max\limits_{\substack{ l_k \in \mathcal{L}', \forall (x,t,i)_k \in O_k \\ (j,t,i) \geqslant (x,t,i)_k , 0 <i \leqslant w} }
\ \ \frac{\widehat{g}_k j l_k(j)}{\widehat{g}_k t l_k(t) + j l_k(j) - t l_k(j+i)}
\eeq

\label{PoAbound:main}
\end{result}


As a corollary from above note that every game $G$ has a $PoA$ lower
bounded by the maximum value of $g_k$ at an ordered triple, i.e  $PoA =
\Omega(g^*_k)$ since the expression above equates to $g^*_k$ when $j =x$ 
for any $(x,t,i)_k \in O_k$.

Theorem~\ref{PoAbound:main} compactly describes a necessary and sufficient condition for the boundedness of the Price of Anarchy of any game
with cost functions drawn from $\mathcal{L}$. To the best of our knowledge, this is the first exact formulation of Price of Anarchy bounds 
(both lower and upper bounds) for unsplittable congestion games with {\it arbitrary} latency cost functions. Previous results by Awerbuch 
{\it et al.} and Monien {\it et al.} \cite{awerbuchSICOMP,monienSICOMP} have provided a tight characterization  for games with polynomial latency cost 
functions. In \cite{roughgardenSTOC09}, Roughgarden provides generalized existence conditions for the Price of Anarchy of congestion games 
using smoothness characterizations.
We demonstrate in this
paper for every superpolynomial cost function in $\mathcal{L}_1$ and $\mathcal{L}_2$, the existence of games for which the $PoA$ is unbounded
and also provide a tight bound on the $PoA$ for all games polynomially bounded latency costs.


\begin{result} The relation between the Price of Anarchy, latency cost functions and the number of players is as follows:

\begin{enumerate}
\item For every superpolynomial cost function in $\mathcal{L}_1$ and $\mathcal{L}_2$, there exist congestion games where the $PoA$ is arbitrarily large 
and increases with the number of players even with bounded weights.
\item
For every polynomially bounded congestion game (with latency functions drawn from $\mathcal{L}_3$), the $PoA$ is independent of the number of players and 
bounded only by the parameters of the cost function.
\end{enumerate}


\label{PoA:iffcondition}
\end{result}

Result~\ref{PoA:iffcondition} directly relates the Price of Anarchy
of unsplittable congestion games to the growth rates of the latency
cost functions that control the player costs in these games. More
significantly, it has strongly negative implications for the Price of
Anarchy of many such games.  These implications were heretofore unknown,
as the only known results to date were on the Price of Anarchy of
congestion games with polynomial cost functions.

Our result implies that the Price of Anarchy is
finite only for those games with latency cost functions bounded by some
polynomial. Latency costs growing faster than polynomial functions have
strongly negative consequences for the Price of Anarchy of every game
with player costs controlled by these functions. For every cost function in this class, there are games 
with {\it unbounded} Price of Anarchy, even if the weights of all players are infinitesimally small.


{\it Remark}: The $PoA$ is bounded for all games with latency functions from $\mathcal{L}_3$ such as the well-known polynomial 
cost functions \cite{awerbuchSICOMP,monienSICOMP} of degree $d \geqslant 0$ described by  $l_k(x) = \sum_{q=0}^d a_q x^q$,. We show in this paper 
that the $PoA$ is bounded even for other polynomial bounded functions, for example, $l_k(x) = (\sum_{q=0}^d a_q x^q)(\sum_{q=0}^m b_q \log^q x)$.

The verdict on the existence of the Price of Anarchy for games with faster growing cost functions is strongly negative. 
This includes both slowly growing super-polynomial cost functions such as $l_k(x) =  a_k x^{\log^{\epsilon} x}$, $\epsilon >0$ 
as well as fast-growing exponential cost functions such as $l_k(x) = x!$ and $l_k(x) = a^x$, where $a > 1$.


%

\section{Game Formulation}
\label{section:definitions}

An {\em unsplittable congestion game} is a strategic game
$G = (\Pi, W, R, \cS, (l_r)_{r \in R})$
where:
\begin{itemize}
\item
$\Pi = \{\pi_1,\ldots, \pi_N\}$ is a non-empty and finite set of players. 

\item  
Player weight set $W \in \mathbb{R}^+$ 
where each player $\pi_i$ has an associated weight $w_{\pi_i} \in W$ 
(also denoted as $i$ later in the analysis),
and the maximum player weight is
$w =\max_{\pi_i \in \Pi} w_{\pi_i}$.

\item
$R = \{r_1,\ldots,r_\zeta\}$ is a non-empty and finite set of resources.

\item
Strategy profile $\cS = \cS_{\pi_1} \times \cS_{\pi_2} \times \cdots \times \cS_{\pi_N}$,
where $\cS_{\pi_i}$ is a {\em strategy set} for player $\pi_i$,
such that $\cS_{\pi_i} \subseteq 2^R$.
Each strategy $S_{\pi_i} \in \cS_{\pi_i}$ is {\em pure} in the sense that it is
a single element from $\cS_{\pi_i}$ 
(in contrast to a {\em mixed strategy} 
which is a probability distribution over the strategy set of the player).
A {\em game state}
is any $S \in \cS$.
We consider {\em finite games} which have finite $\cS$ (finite number of states).
\end{itemize}

In any game state $S$,
let $S_{\pi_i}$ denote the strategy of player $\pi_i$.
We define the following terms with respect to a state $S$:
{\it Congestion}: Each resource $r \in R$ has a congestion $C_r (S) = \sum_{\pi_i \in \Pi \wedge r \in S_{\pi_i}} w_{\pi_i}$,
which is the sum of the weights of the players that use it.
{\it Utility:} In any game state $S$, each resource $r \in R$ has a {\em utility cost} (also referred to as {\em latency cost}) $l_r(S)$.
{\it Player Cost:} In any game state $S$,
each player $\pi \in \Pi_G$ has a {\em player cost} $pc_{\pi_i}(S) = \sum_{r \in S_{\pi_i}} l_r(S)$.
{\it Social Cost:}
In any game state $S$,
the {\em social cost} is 
$SC(S)  = \sum_{r \in R} l_r(C_r(S)) \cdot C_r(S)$.
Note that the social cost is the weighted sum of the player's costs.

%
%
%
%
%
%

When the context is clear, we will drop the dependence on $S$.
For any state $S$,
we use the standard notation
$S = (S_{\pi_i},S_{-{\pi_i}})$
to emphasize the dependence on player $\pi_i$.
Player $\pi_i$
is \emph{locally optimal} (or {\em stable}) in state $S$ if
$pc_{\pi_i}(S) \leq pc_{\pi_i}(S'_{\pi_i},S_{-{\pi_i}})$ for
all strategies $S'_{\pi_i} \in \cS_{\pi_i}$.
A greedy move by a player $\pi_i$ is any change of its strategy
from $S'_{\pi_i}$ to $S_{\pi_i}$
which improves the player's cost,
that is, $pc_{\pi_i}(S_{\pi_i},S_{-{\pi_i}}) < pc_{\pi_i}(S'_{\pi_i},S_{-{\pi_i}})$.
A state $S$ is in a {\em Nash Equilibrium}
if every player is locally optimal,
namely, no greedy move is available for any player.
A Nash Equilibrium realizes the notion of
a stable selfish outcome.
In the games that we study there could exist multiple Nash Equilibria.

A state $S^*$ is called {\em optimal} if it has minimum attainable
social cost: for any other state $S$, $SC(S^*) \le SC(S)$.  
We quantify the quality of the states which are Nash
Equilibria with the \emph{price of anarchy} ($PoA$) (sometimes referred
to as the coordination ratio).  Let $\cal P$ denote the set of distinct
Nash Equilibria.  Then the price of anarchy of game $G$ is:
\begin{equation*}
PoA(G)
=\max\limits_{ S \in {\cal P}} \frac{SC(S)}{SC(S^*)}
\end{equation*}

\section{Preliminaries}

Let $S$ denote an arbitrary (not necessarily an equilibrium) state of
game $G$ with resource set $R$.  We group resources in $R$ together based
on congestion and cost parameters. Let $R_{j,k}^t \subseteq R$ denote an
equivalence class of resources such that for every $r \in R_{j,k}^t$, we
have $C_r(S) = j$, $C_r(S^*) =t$ and latency costs governed by function
$l_k(C_r(S)) \in \mathcal{L}$ i.e the cost of using $r \in R_{j,k}^t$ in
states $S$ and $S^*$ are given by $l_k(j)$ and $l_k(t)$, respectively.
For notational convenience, we label the set of resource equivalence
classes by $\mathcal{E} = \{ R_{j,k}^t \}$.

For any resource $r \in R$, let $\Pi_{r} = \{ \pi | r \in S_{\pi} \}$
and $\Pi^*_{r} = \{ \pi | r \in S^*_{\pi} \}$.  Let $\sigma_i \subseteq
\Pi$ denote the set of players in $\Pi$ with weight $i$, $ 0 < i \leqslant w$
and let $\alpha_{ir} = |\sigma_i \bigcap \Pi^*_r|$ denote the number of
players of weight $i$ utilizing resource $r$ in the optimal state $S^*$.
Thus $\sum_{i: \sigma_i \neq \phi} i \cdot \alpha_{ir} = t$ for all $r
\in R_{j,k}^t$.


For any $t > 0$, let $f(t)$ be the total number of combinations of players
of different weights which can satisfy the equation $\sum_{i: \sigma_i
\neq \phi} i \cdot \alpha_{ir} = t$, where $f(t)$ can be an exponential
function of $t$.  We denote a particular such player combination by
the index $t_a$, $1 \leqslant a \leqslant f(t)$ and let $R^{t_a}_{j,k} \subseteq
R^t_{j,k}$ denote the equivalence class of resources with identical
configurations of players in the optimal state as represented by $t_a$. If
$t=0$, we will sometimes use the notation $0_0$ to denote the empty
configuration.  We represent the optimal configuration in state $S^*$
on any resource $r \in R^{t_a}_{j,k}$, by the vector $\olta = <\{
\alpha_{j,k}^{i,t_a} \}>$, where $\alpha_{j,k}^{i,t_a} > 0$ denotes the
number of players of weight $i$ in configuration $t_a$ on resource $r$
in optimal state $S^*$.  Henceforth we use the notation $i \in \olta$
to denote the presence of a player of weight $i$ in configuration $t_a$,
i.e if $\alpha_{j,k}^{i,t_a} > 0$.

%

We derive our $PoA$ bound by obtaining a constrained maximization formulation of the $PoA$ using resource equivalence classes that can then be bounded.  In particular, consider the term $|R^{t_a}_{j,k}|
i \alpha_{j,k}^{i,t_a} l_k(t)$ representing resource equivalence class
$R_{j,k}^{t_a}$.  This term represents the net contribution of all players of a given weight $i$ occupying the subset of resources $R_{j,k}^{t_a}$  towards the optimal social cost.  More formally define the terms

\bea
\lambda_{j,k}^{i,t_a} &= 
&\frac{|R_{j,k}^{t_a}| \cdot i \cdot \alpha_{j,k}^{i,t_a} l_k(t)}{SC(S^*)},
\quad R_{j,k}^t \in \mathcal{E}, t > 0, 1 \leqslant a \leqslant f(t), i: \sigma_i \in \Pi,
\label{FML:c2} \\
\lambda_{j,k}^{i,0} &= &\frac{|R_{j,k}^0|}{SC(S^*)},
\quad  \quad R_{j,k}^0 \in \mathcal{E}, i: \sigma_i \in \Pi.
\label{FML:c3}
\eea

Each term of Eq.~\ref{FML:c2} represents the fractional net contribution
of players of weight $i$ occupying resources in $R_{j,k}^{t_a}$ with $t>0$
towards the optimal social cost $SC(S^*)$. However as shown in the lemma
below, these terms also represent exactly the contribution of these
players towards the total Price of Anarchy.  Also, Eq.~\ref{FML:c3}
is defined for arbitrary $i$ for consistency with Eq.~\ref{FML:c2},
however in actuality $i$ can be assumed 0 since there are no players of
any weight on resources in $R_{j,k}^0$ in the optimal state.

Denote the coordination ratio of any unsplittable congestion game $G$ as
$H(S) = SC(S)/SC(S^*)$ for arbitrary state $S$ and optimal state $S^*$.
The following lemma relates the coordination ratio to the coefficients
$\lambda_{j,k}^{i,t_a}$.

\begin{lemma}
Given any game state $S$, the coordination ratio $H(S)$ of an
unsplittable congestion game with latency cost functions derived from
class $\mathcal{L}$ can be expressed as

\begin{flalign}
H(S) \quad = 
&\quad \sum_{R_{j,k}^{t_a} \in \mathcal{E}} 
\sum_{i \in \olta} 
\lambda_{j,k}^{i,t_a} \cdot \frac{j}{t} \cdot \frac{l_k(j)}{l_k(t)}
+
\quad \sum_{R_{j,k}^0 \in \mathcal{E}} 
\lambda_{j,k}^{0,0} \cdot j \cdot l_k(j)
\label{H(S):obj} \\
&\mbox{where} \notag \\
&\quad 
\sum_{R_{j,k}^{t_a} \in \mathcal{E}} \sum_{i \in \olta}
\lambda_{j,k}^{i,t_a} \quad = \quad 1
\label{lambda:c2}  \\
&\quad 
\lambda_{j,k}^{0,0} \quad \geqslant  \quad 0
\label{lambda:prop3}  
\end{flalign}

\label{H(S):Formulation}
\end{lemma}
{\it Proof}: Please see Appendix.
%
%
%
%

Let $\mathcal{P}$ denote the set of Nash equilibrium states of $G$.
Then from the definition of the Price of Anarchy, we have
\beq
PoA(G) \quad = \quad \max_{S \in \mathcal{P}} H(S) 
\label{FML:obj} \\
\eeq

Note that for any group of resources $R_{j,k}^t$, the term
$\frac{jl_k(j)}{tl_k(t)}$ represents the localized $PoA$.  Thus given
constraint~\ref{lambda:c2} we can also view the overall $PoA$ of the
game as the average of the localized $PoA$'s on each resource class.
Also note that while the $\lambda_{j,k}^{i,t_a}$, $t > 0$, terms
representing actual optimal player configurations are constrained, the
$\lambda_{j,k}^{0,0}$ terms representing resources contributing to the
equilibrium cost but not the optimal are not.  However as shown later
they cannot be too large.

\section{ Price of Anarchy Lower Bounds}
\label{LB}

 For any arbitrary cost function in $\mathcal{L}$, we describe a specific game which bounds the Price
of Anarchy from below. Consider a game 
$G = (\Pi, W, R, \cS, l_k)$,
where $l_k \in \mathcal{L}$ and 
players $\Pi = \{ \pi_1, \ldots, \pi_{N}\}$
such that every player has a demand of weight exactly $w \in W$.
The set of resources $R$, where $\zeta = |R|$,
can be divided into two disjoint sets $R = A \cup B$,
$A \cap B = \emptyset$, such that
$A = \{a_0, \ldots, a_{\zeta_1 - 1} \}$
and $B = \{b_0, \ldots, b_{\zeta_2 - 1} \}$,
that is, $\zeta = \zeta_1 + \zeta_2$.

Consider parameters $\alpha, \beta, \gamma, \delta \geq 0$,
such that
$\zeta_1 \geq \alpha + \beta$,
and $\zeta_2 \geq \gamma + \delta$. 
Each player $\pi_i$ has two strategies $\cS_{\pi_i} = \{s_i, {\overline s}_i\}$.
Strategy $s_{i}$ occupies $\alpha$ consecutive resources from the set $A$,
$s^A_{i} = \{ a_{(i-1) \mod \zeta_1}, \ldots, a_{(i + \alpha - 1) \mod \zeta_1} \}$,
and $\beta$ consecutive resources from $B$,
$s^B_{i} = \{ b_{(i-1) \mod \zeta_2}, \ldots, b_{(i + \beta - 1) \mod \zeta_2} \}$;
namely, $s_{i} = s^A_{i} \cup s^B_{i}$.
Strategy ${\overline s}_{i}$ 
occupies $\gamma$ consecutive resources from the set $A$
such that the first resource is immediately after the last in $s^A_i$,
${\overline s}^A_{i} = \{ a_{(i + \alpha -1) \mod \zeta_1}, \ldots, a_{(i + \alpha + \gamma - 1) \mod \zeta_1} \}$,
and $\delta$ consecutive resources from $B$,
such that the first resource is immediately after the last in $s^B_i$,
${\overline s}^B_{i} = \{ b_{(i + \beta -1) \mod \zeta_2}, \ldots, b_{(i + \beta + \delta - 1) \mod \zeta_2} \}$;
namely, ${\overline s}_{i} = {\overline s}^A_{i} \cup {\overline s}^B_{i}$.

We consider the game state $S = (s_1,\ldots, s_N)$
which consists of the first strategy of each player,
and game state ${\overline S} = ({\overline s}_1,\ldots, {\overline s}_N)$
which consists of the second strategy of each player.
We take the number of players 
$N = \kappa_1 \zeta_1$ and $N = \kappa_2 \zeta_2$,
for integers $\kappa_1, \kappa_2 \geq 0$.

\begin{lemma}
\label{lemma:LB-equil}
State $S$ is a Nash equilibrium.
\end{lemma}
{\it Proof}: Please see Appendix.

As observed in the proof of Lemma \ref{lemma:LB-equil},
in state $S$ each resource $r \in A$ has congestion equal to $j_1 = C_r(S) = N \alpha w / \zeta_1$,
and each resource $r \in B$ has congestion equal to $j_2 = C_r(S) = N \beta w / \zeta_2$.
Similarly, 
in state ${\overline S}$ each resource $r \in A$ has congestion equal to $t_1 = C_r(S) = N \gamma w / \zeta_1$,
and each resource $r \in B$ has congestion equal to $t_2 = C_r(S) = N \delta w / \zeta_2$.
Similar to the previous section,
we define parameters $\lambda_1, \lambda_2 \geq 0$, 
such that $\lambda_1 + \lambda_2 = 1$ and:
$$ \lambda_1 = \frac {\zeta_1 t_1 l_k(t_1)} 
                     {\zeta_1 t_1 l_k(t_1) + \zeta_2 t_2 l_k(t_2)},
\qquad \qquad
\lambda_2 = \frac {\zeta_2 t_2 l_k(t_2)}
                  {\zeta_1 t_1 l_k(t_1) + \zeta_2 t_2 l_k(t_2)}.
$$
From Lemma \ref{lemma:LB-equil}, we have that $S$ is a Nash equilibrium,
and thus, for any player $i$, $pc_{\pi_i}(S) \leq pc_{\pi_i}(S')$.
With an appropriate choice of the game parameters ($\alpha, \beta, \gamma, \delta, \zeta_1, \zeta_2, \kappa_1, \kappa_2$)
and also by adjusting the 
weight $w$ which is a real number,
we can actually get $pc_{\pi_i}(S) = pc_{\pi_i}(S')$
for each player $\pi_i \in \Pi$.
In other words,
$$
\alpha \cdot l_k(j_1) + \beta \cdot l_k(j_2)
=
\gamma \cdot l_k(j_1 + w) + \delta \cdot l_k(j_2 + w).
$$
Therefore,
$$ 
\zeta_1 j_1 \cdot l_k(j_1) + \zeta_2 j_2 \cdot l_k(j_2)
=
\zeta_1 t_1 \cdot l_k(j_1 + w) + \zeta_2 t_2 \cdot l_k(j_2 + w),
$$
and hence,
$$ 
\zeta_1 j_1 \cdot l_k(j_1) - \zeta_1 t_1 \cdot l_k(j_1 + w)
=
\zeta_2 t_2 \cdot l_k(j_2 + w) - \zeta_2 j_2 \cdot l_k(j_2),
$$
or equivalently,
$$
\lambda_1 
\cdot \frac {j_1 l_k(j_1) - t_1 l_k(j_1 + w)}
            {t_1 l_k(t_1)}
=
\lambda_2
\cdot \frac {t_2 l_k(j_2 + w) - j_2 l_k(j_2)}
            {t_2 l_k(t_2)},
$$
which gives,
\begin{equation}
\label{eqn:lambdas}
\lambda_1 
\cdot \left [ \frac {j_1 l_k(j_1)}
            {t_1 l_k(t_1)}
      - \frac {l_k(j_1 + w)}
              {l_k(t_1)} \right ]
=
\lambda_2
\cdot \left [ \frac {l_k(j_2 + w)}
                    {l_k(t_2)}
              - \frac {j_2 l_k(j_2)}
                      {t_2 l_k(t_2)} \right].
\end{equation}

\begin{lemma}
\label{POALB}
For game $G$, the price of anarchy is bounded by:
$$PoA(G) 
\geqslant \max \left( g^*_k,
\max_{\substack{(j_1,t_1,w) \geqslant (x,t_1,w)_k\\
\forall (x,t_1,w)_k \in O_k}} \ \ 
\frac {{\widehat g}_k j_1 l_k(j_1)}
        {{\widehat g}_k t_1 l_k(t_1) +  j_1 l_k(j_1) - t_1 l_k(j_1 + w)}
\right).$$
\end{lemma}
{\it Proof}: Please see Appendix.

\section{ Price of Anarchy Upper Bounds}
\subsection{Constrained Maximization}

Let $S$ be any Nash equilibrium state of $G$.  We find the upper bound on the $PoA$ via the lemma below in which we convert the unconstrained maximization of Eq.~\ref{H(S):obj} into a constrained version. Consider an arbitrary resource equivalence class $R_{j,k}^{t_a}$ in state $S$ of $G$. For any player of weight $i \in \olta$, $0 < i \leqslant \min(t,w)$, define

\bea
f_{j,k}^{i,t_a} &=
&
\frac{j l_k(j)}{t l_k(t)}
-
\frac{l_k(j+i)}{l_k(t)},
\label{fjk:def} \\
f_{j,k}^{0,0} &= &  j l_k(j) 
\label{fjk0:def} 
\eea 

\no Also let $\gjkita = - \fjkita$. We first define the notion of underloaded and overloaded resource sets conditioned on the value of $\fjkita$. 


\begin{definition}
Resource subset $R_{j,k}^{t_a}$ is defined to be overloaded with respect to players of weight $i$ if $\fjkita \geqslant 0$ and underloaded if $\fjkita < 0$ ($\gjkita > 0$). Define $F = \{ f_{j,k}^{i,t_a} : \fjkita \geqslant 0 \}$  and $T = \{ g_{j,k}^{i,t_a} :  \gjkita > 0 \}$. 
\label{overloaded}
\end{definition}

It can be seen that the first term in the definition of $\fjkita$ above is related to the overall social cost of players using resource class $R_{j,k}^{t_a}$ while the second term is related to the cost of a player of weight $i$ switching to a resource in $R_{j,k}^{t_a}$ from its current strategy. Thus the magnitude of the function $\fjkita$ is an indicator of the contribution of the corresponding resource class $R_{j,k}^{t}$ to the overall Price of Anarchy and also indicates the excess load over the switching costs in that resource class. However for any equilibrium state $S$, the switching costs must exceed the players costs when taken over all resource classes and thus the weight of the overloaded resource classes must be constrained by the underloaded resource classes, as we show through the lemma below.



\begin{lemma}
Let $S \in \mathcal{P}$ be any Nash equilibrium state of game $G$.
Then we must have,
\beq
\sum_{R_{j,k}^0 \in \mathcal{E}} \lambda_{j,k}^{0,0} f_{j,k}^{0,0}
+
\sum_{\substack{R_{j,k}^{t_a} \in \mathcal{E} \\ \sum_{i \in \olta}}} 
\sum_{\fjkita \in F}
\lambda_{j,k}^{i,t_a} f_{j,k}^{i,t_a}
\quad \leqslant \quad 
\sum_{\substack{R_{j,k}^{t_a} \in \mathcal{E} \\ \sum_{i \in \olta}}} 
\sum_{\gjkita \in T} 
\lambda_{j,k}^{i,t_a} g_{j,k}^{i,t_a}
\label{ML:1} \\
\eeq
\label{lemma:mainconstraint}
\end{lemma}
{\it Proof}: Please see Appendix.

Next in order to bound the $PoA$ objective function, we relate the functions $\fjkita$ and $\gjkita$ above to the ordered triples of any latency function $l_k()$. From the definition of ordered triples in Eq.~\ref{orderedtriples}, we obtain

\begin{lemma}
Let $(x,t,i)_k \in O_k$ be any ordered triple of function $l_k()$. Then we must have, 
$(j,t,i) < (x,t,i)_k \in O_k$ for all $\gjkita \in T$ 
and
$(j,t,i) \geqslant (x,t,i)_k \in O_k$ for all $\fjkita \in F$.
\label{ordtriple}
\end{lemma}

Rewriting lemma~\ref{lemma:mainconstraint} using lemma~\ref{ordtriple} above and combining with Eq.~\ref{FML:obj}, the upper bound on the $PoA$ can therefore be expressed as the following constrained maximization:

\begin{definition}[Maximization Problem]
\begin{flalign}
&PoA \quad \leqslant \quad
\max_{S \in \mathcal{P}} H(S) \quad := \quad 
\sum_{R_{j,k}^0 \in \mathcal{E}}
\lambda_{j,k}^{0,0} \cdot j l_k(j) \quad + \notag \\
& \quad 
\sum_{R_{j,k}^{t_a} \in \mathcal{E}} \sum_{i \in \olta} 
\sum_{ (j,t,i) \geqslant (x,t,i)_k }
\lambda_{j,k}^{i,t_a} \cdot  \frac{j l_k(j)}{t l_k(t)}
+
\sum_{R_{j,k}^{t_a} \in \mathcal{E}} \sum_{i \in \olta} 
\sum_{ (j,t,i) < (x,t,i)_k }
\lambda_{j,k}^{i,t_a} \cdot  \frac{j l_k(j)}{t l_k(t)}
\label{PoA:1} \\
&\quad \mbox{s.t}  \notag \\
&\sum_{R_{j,k}^{t_a} \in \mathcal{E}} \sum_{i \in \olta} 
\sum_{ (j,t,i) \geqslant (x,t,i)_k }
\lambda_{j,k}^{i,t_a} f_{j,k}^{i,t_a}
\quad \leqslant \quad 
\sum_{R_{j,k}^{t_a} \in \mathcal{E}} \sum_{i \in \olta} 
\sum_{ (j,t,i) < (x,t,i)_k }
\lambda_{j,k}^{i,t_a} g_{j,k}^{i,t_a}
\label{PoA:2} \\
&\quad 
\sum_{R_{j,k}^{t_a} \in \mathcal{E}} \sum_{i \in \olta}
\lambda_{j,k}^{i,t_a} \quad = \quad 1
\label{PoA:3}  \\
&\quad 
\lambda_{j,k}^{0,0} \quad \geqslant  \quad 0
\label{PoA:4}  
\end{flalign}
\end{definition}


%
%
%

The following lemma provides an exact formulation for evaluating the upper bound on the $PoA$ of any game $G$ 
by bounding the objective function $H(S)$ for any equilibrium state $S$.

\begin{lemma}
\[
PoA(G) \quad \leqslant \quad  \max \left( g^*_k,
\max_{\substack{(j,t,i) \geqslant (x,t,i)_k \\
\forall (x,t,i)_k \in O_k}} \ \ 
\frac{\wg_k j l_k(j)}{\wg_k t l_k(t) + j l_k(j)- t l_k(j+i)} 
\right)
\]
\label{POAUB}
\end{lemma}
{\it Proof}: Please see Appendix.

\section{$PoA$: Functional Characterizations}

Now we can combine the lower and upper bounds for the $PoA$ as derived in
lemma~\ref{POALB} and lemma~\ref{POAUB} to get the tight bounds 
described in Eq.~\ref{PoAbound:main} of Result~\ref{PoAbound:main}.
Next we bound the expression in Eq.~\ref{PoAbound:main} for arbitrary
latency functions $l_k()$.

\begin{theorem}
For every latency function $l_k \in \mathcal{L}_1$, there exist congestion
games with arbitrarily large $PoA$s depending only on the number of players.
\label{l1bound}
\end{theorem}
\begin{proof}
First consider the following example where $i=1$ and $l_k(x+1) >
x l_k(x)$ (for example, latency functions such as $l_k(x) = x!$ or
$l_k(x) = x^x$). Note that in these cases, ordered triples do not exist
since $tl_k(x+1) > x l_k(x)$ for all $t \geqslant 1$. For such functions,
$\wg_k = \max_{x,t \geqslant 1} \left( \frac{l_k(x+i)}{l_k(t)} - \frac{x
l_k(x)}{tl_k(t)} \right)$ is unbounded and therefore so is $g^*_k$. Since
the $PoA \geqslant g^*_k$ it is also unbounded.

Consider a more general example of latency functions in $\mathcal{L}_1$
where ordered triples exist. Given any $i>0$, let $\lim_{x \to
\infty}  \frac{l_k(x+i)}{l_k(x)} = 1 + \delta$, $\delta >0$ is a
constant independent of $t$. Choose an ordered triple $(x,t,i)$ such
that $tl_k(x+i) = xl_k(x)$ and $l_k(x+i) \geqslant (1+\delta)l_k(x)$. Thus $x
\geqslant (1+\delta)t$. Substituting in the expression for the $PoA$ above
with $j=x$, we have the $PoA \geqslant \frac{xl_k(x)}{tl_k(t)}$. Choose $t$
large enough so that

\[
PoA > \frac{l_k\left( (1+\delta)t \right)}{l_k(t)} = 
\frac{l_k(t + \delta t)}{l_k(t + \delta t -i)} \cdot 
\frac{l_k(t + \delta t -i)}{l_k(t + \delta t -2i)} \cdots \frac{l_k(t+i)}{l_k(t)}
\geqslant \left(1 + \delta \right)^{\delta t/i}
\]

From Section~\ref{LB}, since $t$ is controlled by the number of
players in the game which can be arbitrarily large while player weight
$i$ is bounded by a given constant $w$, the $PoA$ is unbounded.
\end{proof}

\begin{theorem}
For every latency function $l_k \in \mathcal{L}_2$, there exist congestion games with arbitrarily large $PoA$ depending only on the number of players.
\label{l2bound}
\end{theorem}
\begin{proof}
Let $t_0 \in \mathbb{R}^+$ be a sufficiently large constant.
Consider ordered triples $(x,t,i)_k \in O_k$ with $t \geqslant t_0$, $x \geqslant t$, for cost function $l_k() \in \mathcal{L}_2$.  We can safely assume that $\frac{l_k(x)}{l_k(t)}$ is bounded for all $t \geqslant t_0$,  else the $PoA$ is unbounded as $g^*_k \geq \frac{xl_k(x)}{tl_k(t)}$ is unbounded. Assume $\frac{l_k(x)}{l_k(t)} \leqslant \kappa$ for  all $t \geqslant t_0$ or equivalently
\beq
\log^{1+\epsilon}x - \log^{1+\epsilon}t \leqslant \log \kappa , \forall t \geqslant t_0
\label{xlx}
\eeq

Also  $\lim_{x \to \infty} l_k(x+i)/l_k(x) = 1$ and $l_k(x+i)/l_k(x) = x/t$, which implies $\exists \epsilon_t \to 0$ such that $x =(1+\epsilon_t)t$.


%
%
%
%

Since $g^*_k$ is assumed bounded and $jl_k(j) > tl_k(j+i)$ for all
$j \geqslant x: (x,t,i)_k \in O_k$, we can bound the $PoA$ expression 
in Eq.~\ref{PoAbound:main}
as

\beq
PoA = \Omega \left(
\max_{\substack{(j,t,i) \geqslant (x,t,i)_k \\ \forall (x,t,i)_k \in O_k}} \ \
\frac{\wg_k tl_k(j+i))}{\wg_k t l_k(t) + j l_k(j)- t l_k(j+i)}
\right)
\label{mainW*}
\eeq

Denote the term above by $y$. Taking the partial derivative of $y$ with
respect to $j$ and equating it to 0 gives us

\beq
\left. \frac{\partial y}{\partial j} \right|_0 \implies
\wg_k t l(t) - jl_k(j) = \left(l_k(j+i) \right) 
\frac{(jl_k(j))'}{tl'_k(j+i)}
\label{dydj:eq}
\eeq
\no where $()'$ denotes the partial derivative with respect to $j$.

For any given value of $t: (x,t,i)_k \in O_k$, $y$ is maximized for $j\geqslant x$ that satisfies Eq.~\ref{dydj:eq}

Substituting this in Eq.~\ref{mainW*} and simplifying we get

\beq
PoA \geqslant 
\max_{\substack{(j,t,i) \geqslant (x,t,i)_k \\ \forall (x,t,i)_k \in O_k}} \ \
\frac{\wg_k }{ \frac{\left(jl_k(j)\right)'}{t \cdot l'_k(j+i)} -1} 
\label{PoAlb:1}
\eeq

%


Let $\alpha_i(j) =l_k'(j+i)/l_k(j+i)$, $i \geqslant 0$. Using the fact that $\alpha_i(x) = (1+\epsilon) \frac{\log^{\epsilon}(x+i)}{x+i}$, $i \geqslant 0$ for $l_k(x) = a_k e^{\log^{1+\epsilon}x}$, we have

\bea
\frac{\left(jl_k(j)\right)'}{t \cdot l'_k(j+i)} \quad - \quad 1 \quad 
&= 
&\quad \frac{\alpha_0(j) + 1/j}{\alpha_i(j)} \cdot \frac{j l_k(j)}{t l_k(j+i)} \quad - \quad 1
\label{PoAlb:2} \nonumber \\
&=
&\quad \left(1 + \frac{i}{j} \right) 
\left( \frac{1 +(1+\epsilon)\log^{\epsilon}j}{(1+\epsilon)\log^{\epsilon}(j+i)} \right) 
\left( \frac{j l_k(j)}{t l_k(j+i)} \right) \quad - \quad 1
\label{PoAlb:3} \nonumber\\
&\leqslant
&\quad \left(1 + \frac{i}{j} \right)
\left(1 + \frac{1}{(1+\epsilon)\log^{\epsilon}j} \right)
\left(\frac{j}{t} \right) \quad - \quad 1
\label{PoAlb:4} 
\eea

Similarly Eq.~\ref{dydj:eq} can be simplified as

\begin{flalign}
&\frac{j l_k(j)}{t l_k(t)} \quad = \quad \frac{\wg_k \cdot \alpha_i(j)}{ \alpha_0(j) - \alpha_i(j) +1/j} \quad \leqslant \quad  \wg_k \frac{1+\epsilon}{1+i/j} \log^{\epsilon}(j+i)  \quad \leqslant \quad \beta \log^{\epsilon}(j+i)
\label{PoAlb:5}
\end{flalign}
where $\beta = (1+\epsilon) \wg_k$ is a constant dependent only on the parameters of latency function $l_k(x)$.

Since $j \geq x$ where $(x,t,i)_k$ is an ordered triple, let $j = \gamma x$ where $\gamma \geqslant 1$ and  $x = (1+\epsilon_t)t$ as defined earlier. Substituting in Eq.~\ref{PoAlb:5}, we get

\begin{flalign}
&\gamma (1+\epsilon_t) e^{\log^{1+\epsilon}(\gamma x) - \log^{1+\epsilon}t} \leqslant \beta \log^{\epsilon}(\gamma x+i) 
\notag \\
&\Rightarrow \log \gamma + (\log \gamma + \log x)^{1+\epsilon} - \log^{1+\epsilon}t \leqslant \log \beta + \epsilon \log \log(2\gamma x) 
\notag \\
&\Rightarrow \log \gamma + \epsilon' \log \gamma \log x + 
log^{1+\epsilon}x - \log^{1+\epsilon}t \leqslant \log \beta + \epsilon \log \log(2\gamma x ) 
\label{PoAlb:6} 
\end{flalign}
\no where $\epsilon'$ is a constant. Further substituting from Eq.~\ref{xlx}, we get,

\begin{flalign}
&\log \gamma(1 + \epsilon' \log x) + \log \kappa
\leqslant \log \beta + \epsilon \log \log(2\gamma x) 
\label{PoAlb:7}\\
&\Rightarrow \log \gamma \leqslant \frac{\zeta + \epsilon \log \log 2\gamma x}{1 + \epsilon' \log x}
\end{flalign}
\no where $\zeta$ is a constant. Since $\epsilon$, $\epsilon'$ and
$\zeta$ are constants and $x$ can be chosen to be $x \gg i$, we have $\log \gamma = \Theta(\frac{\log \log x}{\log x} )$ and so $\gamma =
\Theta(1+\frac{\log \log x}{\log x})$. Substituting for $j/t = (1+\epsilon_t) \gamma$ in Eq.~\ref{PoAlb:4} we notice that all the terms in the first expression on the RHS converge to 1. Further substituting this expression in Eq.~\ref{PoAlb:1} for the $PoA$, we get
$PoA = \Omega((1+\epsilon)\log^{\epsilon}j)$. Since the congestion $j$ depends on the number of players and can be arbitrarily large, we get the result as desired. 
\end{proof}

Finally we consider games from $\mathcal{L}_3$.  \cite{monienSICOMP} describes upper bounds for games with polynomial costs. Here we we present a generalized result for all congestion games with latency functions drawn from the class of polynomially bounded functions.

\begin{theorem}
For every congestion game with latency functions drawn from $\mathcal{L}_3$, the $PoA$ is independent of the number of players and bounded only by the parameters of the cost function (such as the degree of the polynomial).
\label{l3bound}
\end{theorem}
{\it Proof}: Please see Appendix.

%
%
%
%

\section{Conclusions}
We provide the first characterization for the price of anarchy of superpolynomial utilities in congestion games.
We provide tight bounds for a large family of utility functions and show how the price of anarchy increases with the number of players,
while other game parameters such as number of resources and player weights remain fixed.
We also extend and generalize the previously known bounds on games with polynomial utility functions
to games with utility functions inclusive of and bounded above by polynomials.
Our results lead to several interesting open questions:
by restricting player strategy sets and network topologies, can we find interesting families of games with bounded price of anarchy even with superpolynomial utilities?
Another interesting problem is to determine whether there are 
approximate games with bounded $PoA$.

\bibliographystyle{plain}
\bibliography{game}
\clearpage
\appendix
\section{Appendix} 
{\it Proof of lemma~\ref{H(S):Formulation}}:
After substituting the values of $\lambda_{j,k}^{i,t_a}$ from
Eqs.~\ref{FML:c2} and Eq.~\ref{FML:c3} into the coordination
ratio $H(S)$ in Eq.~\ref{H(S):obj}, we get,

\begin{eqnarray*}
H(S) \cdot SC(S^*) \quad & = & \displaystyle
\sum_{R_{j,k}^{t_a} \in \mathcal{E}}  \sum_{i \in \olta} 
|R_{j,k}^{t_a}| i \alpha_{j,k}^{i,t_a} \frac{j \cdot l_k(j)}{t}
+
\sum_{R_{j,k}^0 \in \mathcal{E}} 
|R_{j,k}^{0}| j \cdot l_k(j)
\notag \\
H(S) \cdot SC(S^*)  & = & \sum_{R_{j,k}^{t_a} \in \mathcal{E}} 
\sum_{r \in R_{j,k}^{t_a}}
j \cdot l_k(j)  \frac{\sum_{i \in \olta} i \alpha_{j,k}^{i,t_a}}{t}
+ 
\sum_{R_{j,k}^0 \in \mathcal{E}} 
\sum_{r \in R_{j,k}^{0}}
j \cdot l_k(j) 
\notag \\
\quad & = & \sum_{R_{j,k}^t \in \mathcal{E}}  \sum_{r \in R_{j,k}^t} j \cdot l_k(j)  
\notag \\
\quad & = & \sum_{r \in R} C_r(S) \cdot l_k(C_r(S))
\notag \\
\quad & = & SC(S) 
\notag
\end{eqnarray*}

\no where we use the fact that $\sum_{i \in \olta} i \cdot
\alpha_{j,k}^{i,t_a} =t$ for all values of optimal congestion $t
>0$. Similarly, to prove constraint~\ref{lambda:c2}, note that

\begin{flalign}
\sum_{R_{j,k}^{t_a} \in \mathcal{E}} \sum_{i \in \olta} \lambda_{j,k}^{i,t_a} \cdot SC(S^*) \notag \\
\quad = \sum_{R_{j,k}^{t_a} \in \mathcal{E}} \sum_{i \in \olta}  
|R_{j,k}^{t_a} \cdot i \cdot \alpha_{j,k}^{i,t_a} l_k(t) \notag \\
\quad = \sum_{R_{j,k}^t \in \mathcal{E}} \sum_{r \in R_{j,k}^{t}} t \cdot l_k(t) \notag \\
\quad = SC(S^*) \notag
\end{flalign}

\hfill $\Box$

{\it Proof of lemma~\ref{lemma:LB-equil}}:
Since $N$ is a multiple of $\zeta_1$ and $\zeta_2$, in state $S$ 
each resource in $A$ is utilized by 
$N \alpha / \zeta_1 = \kappa_1 \alpha$ players,
and each resource in $B$ is utilized by 
$N \beta / \zeta_2 = \kappa_2 \beta$ players.
Therefore,
$$pc_{\pi_i}(S) 
= \sum_{r \in s_i} l_k(C_r(S)) 
= \sum_{r \in s^A_i} l_k(C_r(S)) + \sum_{r \in s^B_i} l_k(C_r(S))
= \alpha \cdot l_k(\kappa_1 \alpha w) + \beta \cdot l_k(\kappa_2 \beta w).$$

Let $S' = ({\overline s}_{i}, S_{-\pi_i})$ denote the state derived from $S$ when player $\pi_i$ switches its strategy from $s$ to ${\overline s}$.
Since $s_i \cap {\overline s}_i = \emptyset$,
each resource in $r \in {\overline s}_i$ will 
have congestion $C_r(S') = C_r(S) + w$,
where player $\pi_i$ adds weight $w$ to $r$,
while every other resource in $R$ will have the same congestion in both states.
Consequently,
$$pc_{\pi_i}(S') 
= \sum_{r \in {\overline s}_i} l_k(C_r(S')) 
= \sum_{r \in {\overline s}^A_i} l_k(C_r(S')) + \sum_{r \in {\overline s}^B_i} l_k(C_r(S'))
= \gamma \cdot l_k(\kappa_1 \alpha w + w) + \delta \cdot l_k(\kappa_2 \beta w + w).$$

In order to prove that $S$ is a Nash equilibrium,
it suffices to show that $pc_{\pi_i}(S') - pc_{\pi_i}(S) \geq 0$.
We have,
$$
pc_{\pi_i}(S') - pc_{\pi_i}(S)
=\gamma \cdot l_k(\kappa_1 \alpha w + w) + \delta \cdot l_k(\kappa_2 \beta w + w) - \alpha \cdot l_k(\kappa_1 \alpha w) - \beta \cdot l_k(\kappa_2 \beta w)\nonumber.
$$
Therefore, we only need to show that:
\begin{equation}
\label{eqn:abgd}
\alpha \cdot l_k(\kappa_1 \alpha w) - \gamma \cdot l_k(\kappa_1 \alpha w + w) 
\leq \delta \cdot l_k(\kappa_2 \beta w + w) - \beta \cdot l_k(\kappa_2 \beta w).
\end{equation}
If
$\alpha \cdot l_k(\kappa_1 \alpha w) - \gamma \cdot l_k(\kappa_1 \alpha w + w) \leq 0$, then
by taking $\delta = \beta$,
since $l_k$ is a non-decreasing function,
we get $\delta \cdot l_k(\kappa_2 \beta w + w) - \beta \cdot l_k(\kappa_2 \beta w) \geq 0$;
hence, Eq. \ref{eqn:abgd} holds.

If
$\alpha \cdot l_k(\kappa_1 \alpha w) - \gamma \cdot l_k(\kappa_1 \alpha w + w) > 0$, then
by setting $\beta = \beta' \zeta_2 / \zeta_1$,
and $\delta = \delta' \zeta_2 / \zeta_1$,
for some $\beta', \delta' \geq 0$, 
and we get:
$$
\delta \cdot l_k(\kappa_2 \beta w + w) - \beta \cdot l_k(\kappa_2 \beta w) 
= \delta'  \frac {\zeta_2} {\zeta_1}  \cdot l_k\left(\kappa_2 \beta' \frac {\zeta_2} {\zeta_1} w + w\right) 
- \beta' \frac {\zeta_2} {\zeta_1} \cdot l_k\left(\kappa_2 \beta' \frac {\zeta_2} {\zeta_1} w\right).
$$
Then, Eq. \ref{eqn:abgd},
is equivalent to:
$$
\label{eqn:abgd2}
\zeta_1 \left(\alpha \cdot l_k(\kappa_1 \alpha w) - \gamma \cdot l_k(\kappa_1 \alpha w + w)\right) 
\leq \zeta_2 \left(\delta' \cdot l_k\left(\kappa_2 \beta' \frac {\zeta_2} {\zeta_1} w + w\right) - \beta' \cdot l_k\left(\kappa_2 \beta' \frac {\zeta_2} {\zeta_1} w\right)\right).
$$
For taking $\zeta_2 \geq \kappa_1 \alpha \zeta_1$,
and by setting $\delta'$ and $\beta'$ such that $\delta' - \beta' \geq (\alpha - \gamma)/(\kappa_1 \alpha)$,
we get:
\begin{eqnarray}
\zeta_2 \left(\delta' \cdot l_k\left(\kappa_2 \beta' \frac {\zeta_2} {\zeta_1} w + w\right) 
- \beta' \cdot l_k\left(\kappa_2 \beta' \frac {\zeta_2} {\zeta_1} w\right)\right)
& \geq & \zeta_2 (\delta' - \beta') l_k \left(\kappa_2 \beta' \frac {\zeta_2} {\zeta_1} w\right)\nonumber\\
& \geq & \zeta_1 (\delta' - \beta') l_k(\kappa_2 \beta' \kappa_1 \alpha w)\nonumber\\
& \geq & \zeta_1 (\alpha - \gamma) l_k(\kappa_1 a w + w)\nonumber\\
& \geq & \zeta_1 \left(\alpha \cdot l_k(\kappa_1 \alpha w) - \gamma \cdot l_k(\kappa_1 \alpha w + w)\right),\nonumber
\end{eqnarray}
as needed.
\hfill $\Box$

{\it Proof of lemma~\ref{POALB}}:
From Lemma \ref{lemma:LB-equil}, state $S$ is a Nash equilibrium.
Therefore, 
\begin{eqnarray}
PoA(G) 
& \geq & \frac {SC(S)} {SC({\overline S})} \nonumber
 =    \frac {\sum_{r \in R} l_k(C_r(S))} 
               {\sum_{r \in R} l_k(C_r(\overline S))} \nonumber\\
& = &    \frac {\sum_{r \in A} j_1 l_k(j_1) + \sum_{r \in B} j_2 l_k(j_2)} 
               {\sum_{r \in A} t_1 l_k(t_1) + \sum_{r \in B} t_2 l_k(t_2)} \nonumber\\
& = & \frac {\zeta_1 j_1 l_k(j_1) + \zeta_2 j_2 l_k(j_2)}
            {\zeta_1 t_1 l_k(t_1) + \zeta_2 t_2 l_k(t_2)} \nonumber\\
& = & \frac {\zeta_1 t_1 l_k(t_1)}
               {\zeta_1 t_1 l_k(t_1) + \zeta_2 t_2 l_k(t_2)}
          \cdot \frac {j_1 l_k(j_1)} { t_1 l_k(t_1)}
        + \frac {\zeta_2 t_2 l_k(t_2)}
                {\zeta_1 t_1 l_k(t_1) + \zeta_2 t_2 l_k(t_2)}
          \cdot \frac {j_2 l_k(j_2)} { t_2 l_k(t_2)}\nonumber\\
& = & \lambda_1 \cdot \frac {j_1 l_k(j_1)} { t_1 l_k(t_1)}
      + \lambda_2 \cdot \frac {j_2 l_k(j_2)} { t_2 l_k(t_2)}\nonumber\\
& \geq & \lambda_1 \cdot \frac {j_1 l_k(j_1)} { t_1 l_k(t_1)}.\nonumber
\end{eqnarray}

Since $\lambda_1 + \lambda_2 = 1$, from Eq. \ref{eqn:lambdas} we can get that
$$ \lambda_1 F = (1 - \lambda_1) {\widehat g}_k, $$
where
$$F =  \frac {j_1 l_k(j_1)}
            {t_1 l_k(t_1)}
      - \frac {l_k(j_1 + w)}
              {l_k(t_1)}, $$
and ${\widehat g}_k$ is obtained from Eq. \ref{ghatk}
such that it maximizes          
$$\frac {l_k(j_2 + w)}
                    {l_k(t_2)}
              - \frac {j_2 l_k(j_2)}
                      {t_2 l_k(t_2)}.$$
Consequently,
$$
\lambda_1
= \frac {{\widehat g}_k} {F + {\widehat g}_k}.
$$
Therefore,
\begin{equation}
PoA(G)
\geq
\lambda_1 \cdot \frac {j_1 l_k(j_1)} { t_1 l_k(t_1)}
= \frac {{\widehat g}_k j_1 l_k(j_1)}
        {{\widehat g}_k t_1 l_k(t_1) +  j_1 l_k(j_1) - t_1 l_k(j_1 + w)}.
\label{lb-lb1}
\end{equation}
A lower bound on the price of anarchy 
follows by considering all ordered triplets of the form $(j_1, t_1, w)$ 
that maximize the right hand in Eq. \ref{lb-lb1}.
A second lower bound for the price of anarchy is $g^*_k$,
defined in Eq. \ref{g*k}, for the case where $\lambda_1 = 1$.
\hfill $\Box$

{\it Proof of lemma~\ref{lemma:mainconstraint}}:
Substituting above for $\lambda_{j,k}^{i,t_a}$, $\gjkiqta$ and $\fjkiqta$ from Eqs.~\ref{FML:c2} and ~\ref{fjk:def} and simplifying, we need to prove:

\begin{flalign}
\frac{ \displaystyle 
\sum_{R_{j,k}^t \in \mathcal{E}} \sum_{a=1}^{f(t)} 
|R_{j,k}^{t_a}| 
\sum_{i \in \olta} 
i \cdot \alpha_{j,k}^{i,t_a} \Big(j l_k(j)/t - l_k(j+i) \Big) 
+
\sum_{R_{j,k}^0 \in \mathcal{E}} |R_{j,k}^{0}| \cdot j l_k(j)
}{SC(S^*)}
\quad \leqslant \quad 0 
\label{c2:1} 
\end{flalign}


Since $SC(S^*) > 0$, consider the numerator.
We use the following simple observation 

\beq
\sum_{r \in R_{j,k}^t} 
\sum_{i: \pi \in \Pi_r \land \pi \in \sigma_i}  i \cdot l_k(j)
=
|R_{j,k}^{t}| \cdot j \cdot l_k(j)
\quad \forall R_{j,k}^t \in \mathcal{E}
\label{Rjk:explain2} 
\eeq

\no since each player $\pi \in \sigma_i$ contributes $i$ towards the
equilibrium congestion value $j$ of every resource $r \in R_{j,k}^t$
that is contained in its equilibrium strategy $S_{\pi}$ and each such
resource contributes $l_k(j)$ towards its player cost.

\begin{flalign}
& \sum_{R_{j,k}^t \in \mathcal{E}} \sum_{a=1}^{f(t)} 
|R_{j,k}^{t_a}| \cdot j \cdot l_k(j)
\sum_{i \in \olta} \frac{i \alpha_{j,k}^{i,t_a}}{t}
+
\sum_{R_{j,k}^0 \in \mathcal{E}} |R_{j,k}^{0}| \cdot j \cdot l_k(j)
\notag \\
&\quad \quad \quad 
-
\sum_{R_{j,k}^t \in \mathcal{E}} \sum_{a=1}^{f(t)} \sum_{i \in \olta} 
|R_{j,k}^{t_a}| i \alpha_{j,k}^{i,t_a}  l_k(j+i) 
\notag \\
& \quad \equiv
\sum_{R_{j,k}^t \in \mathcal{E}} 
|R_{j,k}^{t}| \cdot j \cdot l_k(j)
-
\sum_{R_{j,k}^t \in \mathcal{E}} \sum_{a=1}^{f(t)} \sum_{i \in \olta}
|R_{j,k}^{t_a}| \cdot i \alpha_{j,k}^{i,t_a} l_k(j + i)
\notag \\
& \quad =
\sum_{R_{j,k}^t \in \mathcal{E}}  \sum_{r \in R_{j,k}^t} 
\left(
\sum_{i: \pi \in \Pi_r \land \pi \in \sigma_i}  i \cdot l_k(j)
-
\sum_{i: \pi \in \Pi^*_r \land \pi \in \sigma_i}  
i \cdot l_k(j + i)
\right)
\label{c2:4.2} \\
&=
\sum_{r \in R} \ \
\sum_{i: \pi \in \Pi_r \land \pi \in \sigma_i}  
i \cdot l_r ( C_r )
-
\sum_{r \in R} \ \
\sum_{i: \pi \in \Pi^*_r \land  \pi \in \sigma_i}
i \cdot l_r( C_r + i )
\label{c2:5} \\
&=
\sum_{0<i \leqslant w} \sum_{\pi \in \sigma_i} i \cdot  pc_{\pi}(S_{\pi}, S_{-\pi}) 
-
\sum_{0<i \leqslant w}
\sum_{\pi \in \sigma_i} i \cdot  pc_{\pi}(S^*_{\pi}, S_{-\pi}) 
\label{c2:7} \\
&\leqslant 0
\label{c2:8} 
\end{flalign}

\no where the first term of Eq.~\ref{c2:4.2} uses Eq.~\ref{Rjk:explain2}.
The second term follows from the fact that for each resource $r \in
R_{j,k}^{t_a}$ there are exactly $\alpha_{j,k}^{i,t_a}$ players of
weight $i$ that contain $r$ in their optimal strategies and $l_k(j +
i)$ represents the cost to each such player of switching to resource
$r$ while all other players remain in state $S$.  The left term of
Eq.~\ref{c2:5} represents $i$ times the cost to a player of weight $i$
of any resource in $R$ in state $S$ while the right term represents $i$
times the switching cost to any resource in $R_{j,k}^t$ summed up over
all resources.  Eq.~\ref{c2:7} represents the summation of player costs
in state $S$ and the switching cost to state $S^*$ over all resources
and then over all players.  Finally Eq.~\ref{c2:8} follows since $S$
is a Nash equilibrium. 
\hfill $\Box$ 

From the definition of $g^*_k$ in Eq.~\ref{g*k} and using 
$j l_k(j) < tl_k(j+i)$ for all $(j,t,i) < (x,t,i)_k$ along with
constraint~\ref{PoA:3}, we have

\begin{lemma}
\[
\sum_{R_{j,k}^{t_a} \in \mathcal{E}} \sum_{i \in \olta} 
\sum_{ (j,t,i) < (x,t,i)_k }
\lambda_{j,k}^{i,t_a} \cdot  \frac{j l_k(j)}{t l_k(t)}
\quad \leqslant \quad
 g^*_k 
 \]
\label{ginTbound}
\end{lemma}

Also from the LHS of lemma~\ref{lemma:mainconstraint} and the definition of $\wg$ in Eq.~\ref{ghatk}, we get

\begin{lemma}
\[
\sum_{R_{j,k}^0 \in \mathcal{E}} \lambda_{j,k}^{0,0} \cdot j l_k(j)
\quad \leqslant \quad
\wg_k
\label{rjk0bound}
\]
\end{lemma}

{\it Proof of lemma~\ref{POAUB}}:
Applying lemmas~\ref{ginTbound} and ~\ref{rjk0bound} to the objective function $H(S)$ in ~\ref{PoA:1} we have

\beqs
H(S) \quad \leqslant \quad 
\wg_k + g^*_k +
\sum_{R_{j,k}^{t_a} \in \mathcal{E}} \sum_{i \in \olta} 
\sum_{ (j,t,i) \geqslant (x,t,i)_k }
\lambda_{j,k}^{i,t_a} \cdot  \frac{j l_k(j)}{t l_k(t)}
\eeqs

For every ordered triple $(x,t,i)_k \in O_k$ and
any $(j,t,i) \geqslant (x,t,i)_k$ define 
$\hljkita = \widehat{g}_k / (\widehat{g}_k + \fjkita)$ and
$\hWjkita = \hljkita \frac{j l_k(j)}{t l_k(t)}$. Letting
$W^* = \max_{(j,t,i) \geqslant (x,t,i)_k} \hWjkita$ we can rewrite the expression above as,


\beqs
H(S) \quad \leqslant \quad 
\widehat{g}_k + g^*_k + W^* 
\sum_{R_{j,k}^{t_a} \in \mathcal{E}} \sum_{i \in \olta} 
\sum_{ (j,t,i) \geqslant (x,t,i)_k }
\ljkita \left( 1 + \frac{\fjkita}{\widehat{g}_k} \right)
\eeqs

Using the RHS of lemma~\ref{lemma:mainconstraint} to bound the last term above, we get

\bea
H(S) &\leqslant &
\widehat{g}_k + g^*_k + W^* 
\left(
\sum_{\substack{R_{j,k}^{t_a} \in \mathcal{E} \\ i \in \olta}} 
\sum_{ (j,t,i) \geqslant (x,t,i)_k }
\ljkita 
+ 
\frac{1}{\widehat{g}_k}
\sum_{\substack{R_{j,k}^{t_a} \in \mathcal{E} \\ i \in \olta}} 
\sum_{ (j,t,i) < (x,t,i)_k }
\ljkita \gjkita 
\right)
\notag \\
&\leqslant &\widehat{g}_k + g^*_k +  W^*
\sum_{R_{j,k}^{t_a} \in \mathcal{E}}  \sum_{i \in \olta}
\left(
\sum_{ (j,t,i) \geqslant (x,t,i)_k }
\ljkita 
+
\sum_{ (j,t,i) < (x,t,i)_k }
\ljkita
\right)
\notag \\
&\leqslant &\widehat{g}_k + g^*_k + W^* 
\notag \\
&= &\widehat{g}_k + g^*_k + 
\max_{\substack{(j,t,i) \geqslant (x,t,i)_k \\
\forall (x,t,i)_k \in O_k}} \ \ 
\frac{\wg_k j l_k(j)}
{\wg_k t l_k(t) + jl_k(j) - t l_k(j+i)}
\label{w*} 
\eea

Note that if ordered triples exist for the cost function $l_k()$, 
then since $j^* l_k(j^*) = t^* l_k(j^* + i^*)$ and $\frac{j^*l_k(j^*)}{t^*l_k(t^*)} = g^*_k$ and $(j,t,i) \geqslant (j^*,t^*,i^*)_k$ we must have $\max_{(j,t,i)\geqslant (x,t,i)_k, \forall (x,t,i)_k \in O_k} W^*\geqslant g^*_k \geqslant \wg_k$ and so the last term dominates. However for rapidly growing cost functions, ordered triples may not exist and  so only the first two terms above are relevant. This leads to the expression in the lemma as desired.
\hfill $\Box$

{\it Proof of lemma~\ref{l2bound}}:
Let $z$ denote the $PoA$ expression 
in Eq.~\ref{PoAbound:main}.
Taking the partial derivative $\frac{\partial y}{\partial j}$  and equating to $0$ gives us

\beq
\left. \frac{\partial y}{\partial j} \right|_0 \implies
\wg_k l_k(t) =  l_k(j+i)  -j l_k(j) \frac{l'(j+i)}{(jl_k(j))'}
\label{dydj2}
\eeq

Similarly evaluating the partial with respect to $t$ gives
\beq
\frac{\partial y}{\partial t} = \alpha \left( l_k(j+i) - \wg_k (tl_k(t))' \right) 
\label{dydt2}
\eeq
where $\alpha >0$. Evaluating the two together it can be shown that $\frac{\partial y}{\partial t}$ is decreasing in $t$ and $t=i$ at the maximum value of $z$. Given $0 <i \leq w$, the bound on the $PoA$ can then be obtained by solving the following minimization using standard KKT conditions: 

\begin{flalign}
&\min \left(
\wg_k / \left( 1 - \frac{t l'_k(j+i)}{(j l_k(j))'} \right)
\right)
\label{POAminKKT} \\
& \mbox{ s.t } \notag \\
& \wg_k  l(t) = l_k(j+i) - j l_k(j) \frac{l'(j+i)}{(jl_k(j))'} 
\\
&j l_k(j) \geq t l_k(j+i) \\
&0 <i \leq w
\end{flalign}

\hfill $\Box$

\end{document}